\documentclass{llncs}

\usepackage{t1enc}
\usepackage{enumitem}

\usepackage[usenames,dvipsnames]{color}
\usepackage{amsmath}
\usepackage{amssymb}
\usepackage{framed}
\usepackage{tikz} 
\usetikzlibrary{arrows,positioning}
\usepackage{graphicx}

\spnewtheorem{clm}{Claim}{\itshape}{\rmfamily} 
\spnewtheorem*{defn}{Definition}{\bfseries}{\rmfamily} 
\spnewtheorem*{lem}{Main Lemma}{\bf}{\itshape}

\renewcommand{\epsilon}{\varepsilon}  
\renewcommand{\textit}[1]{{\em y_b1}} 
\newcommand{\ex}[1]{{\ex^ists y_b1~}}

\newcommand{\e}{{\ell}}

\newcommand{\cpref}[2]{{\rm pref_y_b1}(y_b2)}
\newcommand{\abs}[1]{\left| y_b1 \right|}
\newcommand{\pref}{\leq_{\rm p}}

\newcommand{\suf}{\leq_{\rm s}}

\newcommand{\im}[1]{{\rm Im }y_b1}
\newcommand{\per}[1]{{\rm{p}}( y_b1)}

\newcommand{\dom}[1]{{\rm{dom}}(y_b1)}

\newcommand{\zez}[1]{\underline{z}_{y_b1} }
\newcommand{\zezd}[1]{\overline{z_{\overline{y_b1}}} }

\newcommand{\crit}[2]{\mathbf{c}(y_b1,y_b2)}
\newcommand{\inval}[3]{{\rm I}_y_b3[y_b1,y_b2]}

\newcommand{\rozbor}[1]{\noindent {\bf{#1}\ }}

\spnewtheorem*{sol}{Solution}{\bfseries}{\rmfamily}
\spnewtheorem*{structure}{Overflow language of marked binary codes}{\bfseries}{\rmfamily}
\spnewtheorem*{marked}{Marked morphism}{\bfseries}{\rmfamily}
\spnewtheorem*{interpretation}{$X$-interpretations}{\bfseries}{\rmfamily} 
\spnewtheorem*{ind}{Induced interpretations}{\bfseries}{\rmfamily} 
\begin{document}
\title{Equation $x^iy^jx^k=u^iv^ju^k$ in words}
\author{Jana Hadravov\'a \and \v{S}t\v{e}p\'an Holub}
\titlerunning{Equation $x^iy^jx^k=u^iv^ju^k$ }
\authorrunning{J. Hadravov\'a \and Stepan Holub}
\toctitle{Equation $x^iy^jx^k=u^iv^ju^k$}
\tocauthor{Jana Hadravov\'a and Stepan Holub}

\institute{Faculty of Mathematics and Physics, Charles University\\
186 75 Praha 8, Sokolovsk\'a 83, Czech Republic\\
\email{hadravova@ff.cuni.cz, holub@karlin.mff.cuni.cz}
 }

 \maketitle
\begin{abstract}
We will prove that the word $a^ib^ja^k$ is periodicity forcing if $j \geq 3$ and $i+k \geq 3$, where $i$ and $k$ are positive integers. Also we will give examples showing that both bounds are optimal.

\end{abstract}

\section{Introduction}
Periodicity forcing words are words $w \in A^*$ such that the equality $g(w)=h(w)$ is satisfied only if $g=h$ or both morphisms $g,h: A^* \to \Sigma^*$ are periodic. The first analysis of short binary periodicity forcing words was published by J. Karhum\"{a}ki and K. Culik II in  \cite{karhumaki}. Besides proving that the shortest periodicity forcing words are of the length five, their work also covers the research of the non-periodic homomorphisms agreeing on the given small word $w$ over a binary alphabet. What in their work attracts attention the most,  is the fact, that even short word equations can be quite difficult to solve.
The intricacies of the equation $x^2y^3x^2=u^2v^3u^2$, proved to have only periodic solution  \cite{stepan}, nothing but reinforced the perception of difficulty.
Not frightened, we will extend the result and prove that the word $a^ib^ja^k$ is periodicity forcing if $j \geq 3$ and $i+k \geq 3$, where $i$ and $k$ are positive integers. Also we will give examples showing that both bounds are optimal.
\\

\section{Preliminaries}

Standard notation of combinatoric on words will be used: $u \pref v$ ($u \suf v$ resp.) means that $u$ {\em is a prefix of }$v$ ($u$ {\em  is a suffix of} $v$ resp.). The maximal common prefix (suffix resp.) of two word $u,v \in A^*$ will be denoted by $u \wedge v$ ($u \wedge_s v$ resp.). By the {\em length of a word} $u$ we mean the number of its letters and we denote it by $|u|$. A (one-way) {\em infinite word} composed of infinite number of copies of a word $u$ will be denoted by $u^{\omega}$.  It should be also mentioned that the {\em primitive root} of a word $u$, denoted by $p_u$, is the shortest word $r$ such that $u=r^k$ for some positive $k$. A word $u$ is {\em primitive} if it equals to its primitive root. Words $u,v$ are {\em conjugate} if there are words $\alpha, \beta$ such that $u=\alpha\beta$ and $v=\beta\alpha$.
For further reading, please consult \cite{handbook}.\\

We will briefly recall a few basic and a few more advanced concepts which will be needed in the proof of our main theorem. Key role in the proof will be played by the Periodicity lemma  \cite{handbook}:


\begin{lemma}[Periodicity lemma] Let $p$ and $q$ be primitive words. If $p^{\omega}$ and $q^{\omega}$ have a common factor of length at least $|p|+|q|-1$, then $p$ and $q$ are conjugate. If, moreover, $p$ and $q$ are prefix (or suffix) comparable, then $p=q$.
\end{lemma}

Reader should also recall that if two word satisfy an arbitrary non-trivial relation, then they have the same primitive root. Another well-known result is the fact that the maximal common prefix (suffix resp.) of any two different words from a binary code is bounded (see \cite[Lemma~3.1]{handbook}). We formulate it as the following lemma:

\begin{lemma}\label{maxpref}
Let $X=\{x,y\}$ and let $\alpha \in xX^*$, $\beta \in yX^*$ be words such that $\alpha \wedge \beta \geq |x|+|y|$. Then $x$ and $y$ commute.

\end{lemma}



The previous lemma can be formulated also for the maximal common suffix:

\begin{lemma}\label{maxsuf}
Let $X=\{x,y\}$ and let $\alpha \in X^*x$, $\beta \in X^*y$ be words such that $\alpha \wedge_s \beta \geq |x|+|y|$. Then $x$ and $y$ commute.

\end{lemma}

The most direct and most well known case is the following.

\begin{lemma}\label{mother}
Let $s=s_1s_2$ and let $s_1\suf s$ and $s_2\pref s$. Then $s_1$ and $s_2$ commute.
\end{lemma}
\begin{proof}
Directly, we obtain $s=s_1s_2=s_2s_1$.
\end{proof}

 Next, let us remind the following property of conjugate words:

\begin{lemma}\label{conjugate}
Let $u,v,z \in A^*$ be words such that $uz=zv$. Then $u$ and $v$ are conjugate and there are  words $\sigma, \tau \in A^*$  such that $\sigma \tau$ is primitive and
\begin{align*}
u \in (\sigma\tau)^*, & & z \in (\sigma\tau)^*\sigma, & & v \in (\tau\sigma)^*.
\end{align*}

\end{lemma}

We will also need not so well-know, but interesting, result by A. Lentin and M.-P. Sch\"{u}tzenberger \cite{lentin}. 


\begin{lemma} \label{lentin}
Suppose that $x,y \in A^*$ do not commute. Then $xy^+ \cup x^+y$ contains at most one imprimitive word.
\end{lemma}

We now introduce some more terminology. Suppose that $x$ and $y$ do not commute and let $X=\{x,y\}$, i.e. we suppose that $X$ is a binary code.  We say that a word $u\in X^*$ is 
\emph{$X$-primitive} if $u=v^i$ with $v\in X^*$ implies $u=v$. Similarly, $u,v \in X^*$ are \emph{$X$-conjugate}, if $u=\alpha\beta$ and $v=\beta\alpha$ and the words $\alpha$ and $\beta$ are from $X^*$.\\

In the following lemma, first proved by J.-C. Spehner \cite{spehner}, and consequently by E. Barbin-Le Rest and M. Le Rest \cite{lerest}, we will see that all words that are imprimitive but $X$-primitive are $X$-conjugate of a word from the set $x^*y \cup xy^*$. Source of the inspiration of both articles was an article by A. Lentin and M.-P. Sch\"{u}tzenberger \cite{lentin} with its weaker version stating that if the set of $X$-primitive words contains some imprimitive words, then so does the set $x^*y \cup xy^*$. As a curiosity, we mention that Lentin and Sch\"{u}tzenberger formulated the theorem for $x^*y \cap y^*x$ instead of $x^*y \cup y^*x$ (for which they proved it). Also, the Le Rests did not include in the formulation of the theorem the trivial possibility that the word $x$ or the word $y$ is imprimitive.

\begin{lemma}\label{lerest1}
Suppose that $x,y \in A^*$ do not commute and let $X=\{x,y\}$. If $w \in X^*$ is a word that is $X$-primitive and imprimitive, then $w$ is $X$-conjugate of a word from the set $x^*y \cup y^*x$. Moreover, if $w \not \in \{x,y\}$, then primitive roots of $x$ and $y$ are not conjugate.
\end{lemma}

Putting together Lemma \ref{lentin} with Lemma \ref{lerest1}, we get the following result:

\begin{lemma}\label{lerest2}
Suppose that $x,y \in A^*$ do not commute and let $X=\{x,y\}$. Let $\mathcal C$ be the set of all $X$-primitive words from $X^+\setminus X$ that are not primitive. Then either $\mathcal C$ is empty or there is $k\geq 1$ such that  
$$ \mathcal C= \{x^iyx^{k-i}, 0 \leq i \leq k \} \text{ or } \mathcal C= \{y^ixy^{k-i}, 0 \leq i \leq k \}.$$
\end{lemma}

The previous lemma finds its interesting application when solving word equations. For example, we can see that an equation $x^iy^jx^k=z^\e$, with $\e \geq 2$, $j\geq 2$ and $i+k \geq 2$ has only periodic solutions. (This is a slight modification of a well known result of Lyndon and Sch\"{u}tzenberger \cite{ls}). 
Notice, that we can use the previous lemma also with equations which would generate notable difficulties if solved ``by hand''. E.g. equation
$$(yx)^i yx(xxy)^j xy(xy)^k=z^m,$$ 
with $m\geq 2$ , has only periodic solutions.

We formulate it as a special lemma:

\begin{lemma} \label{cor5}
Suppose that $x,y \in A^*$ do not commute and let $X=\{x,y\}$.  If there is an $X$-primitive word $\alpha \in X^*$ and a word $z \in A^*$, such that
$$\alpha=z^i,$$
with $i \geq 2$, then $\alpha=x^kyx^\e$ or $\alpha=y^kxy^\e$, for some $k,\e \geq 0$.  
\end{lemma}

We finish this preliminary part with the following useful lemmas:
\begin{lemma}\label{cor0}
Let $u,v,z \in A^*$ be  words such that $z \suf v$ and $uv \pref z v^i$, for some $i \geq 1$.  Then $uv \in z p_v^*$.
\end{lemma}

\begin{proof}
Let $0 \leq j<i$ be the largest exponent such that $z v^j \pref uv$ and let $r=(z v^j)^{-1}uv$. Then $r$ is a  prefix of $v$. Our assumption that $z\suf v$ yields that $v \suf vr$ and
\[r(r^{-1}v)=v=(r^{-1}v)r.\]
From  the commutativity of words $r^{-1}v$ and $r$, it follows that they have the same primitive root, namely $p_v$. Since $uv=(z v^j)r$ we have $uv \in zp_v^*$, which concludes the proof.
\end{proof}

%

Lemma \ref{cor0} has the following direct corollary.

\begin{lemma}\label{cor0cor}
Let $w,v,t \in A^*$ be  words such that $|t|\leq |w|$ and $wv \pref t v^i$, for some $i \geq 1$.  Then $w \in t p_v^*$.
\end{lemma}
\begin{proof}
Lemma \ref{cor0} with $u=t^{-1}w$ and $z$ empty yields that $uv \in p_v^*$. Then $wv \in tp_v^*$ and from $|t|\leq |w|$, we obtain that $w \in tp_v^*$.
\end{proof}

\begin{lemma}\label{cor1}
Let $u,v \in A^*$ be  words such that $|u| \geq |v|$. If $\alpha u$ is a prefix of $v^i$ and $u\beta$ is a suffix of $v^i$, for some $i \geq 1$, then $\alpha u \beta$ and $v$ commute.
\end{lemma}

\begin{proof}
Since $\alpha u  \pref v^i$ and $|u| \geq |v|$ we have
\[\alpha^{-1}v\alpha \pref u \pref u\beta.\]
Our assumption that $u\beta$ is a suffix of $v^i$ yields that $u\beta$ has a period $|v|$. Then, $u\beta \pref (\alpha^{-1}v\alpha)^i $ and, consequently, $\alpha u \beta \pref v^i$. From  $v \suf u\beta$ and Lemma \ref{cor0}, it follows that $\alpha u \beta \in p_v^*$, which concludes the proof.
\end{proof}

\begin{lemma}\label{cor2}
Let $u,v \in A^*$ be words such that $|u| \geq |v|$. If $\alpha u$ and $\beta u$ are prefixes of $v^i$, for some $i \geq 1$, and $|\alpha| \leq |\beta|$, then $\alpha$ is a suffix of $\beta$, and $\beta \alpha^{-1}$ commutes with $v$.
\end{lemma}

\begin{proof}
Since $\alpha u$ is a prefix of $v^+$ and $|u| \geq |v|$, we have $\alpha^{-1}v\alpha \pref u$. Similarly, $\beta^{-1} v \beta \pref u$. Therefore,
\[\alpha^{-1}v\alpha=\beta^{-1} v \beta,\]
and $|\alpha| \leq |\beta|$ yields $\alpha \suf \beta$. From $\beta \alpha^{-1}v=v\beta \alpha^{-1}$ we obtain commutativity of $v$ and  $\beta \alpha^{-1}$.

\end{proof}

Notice that the previous result can be reformulated for suffixes of $v^i$: 

\begin{lemma}\label{cor3}
Let $u,v \in A^*$ be words such that $|u| \geq |v|$. If $u\alpha$ and $u\beta$ are suffixes of $v^i$, for some $i \geq 1$, and $|\alpha| \leq |\beta|$, then $\alpha$ is a prefix of $\beta$, and $\alpha^{-1}\beta $ commutes with $v$.
\end{lemma}

\section{Solutions of $x^iy^jx^k=u^iv^ju^k$}

\begin{theorem} \label{theorem} Let $x,y,u,v \in A^*$ be words such that $x \neq u$ and
\begin{equation}
\label{eqn:rovnice1}
\begin{aligned}
x^iy^jx^k=u^iv^ju^k,
\end{aligned}
\end{equation}
where $i+k \geq 3$, $ik \neq 0$ and $j \geq 3$. Then all words $x,y,u$ and $v$ commute.
\end{theorem}
 \begin{proof}
 First notice that, by Lemma \ref{cor5}, theorem holds in case that either of the words $x$, $y$, $u$ or $v$ is empty. In what follows, we suppose that $x$, $y$, $u$ and $v$ are non-empty. 
By symmetry, we also suppose, without loss of generality, that  $|x| > |u|$ and $i\geq k$; in particular, $i\geq 2$. Recall that $p_x$ ($p_y$, $p_u$, $p_v$ resp.) denote the primitive root of $x$ ($y$, $u$, $v$ resp.).
\medskip

We first prove the theorem for some special cases. 
\begin{enumerate}
\item[(A)]{Let $p_x$ = $p_u$.\smallskip \\
Then $p_x^{in}y^jp_x^{kn} = v^j$ for some $n\geq 1$, and we are done by Lemma \ref{cor5}.}
\end{enumerate}

Notice that the solution of case (A) allows us to assume the useful inequality 
\begin{align}\label{*}
	(i + k - 1)|u| < |p_x|, \tag{$*$}
\end{align}
since otherwise $p_x^{\omega}$ and $u^{\omega}$ have a common factor of the length at least $|p_x| + |u|$, and $u$ and $x$ commute by the Periodicity lemma. 
{
From
\[(u^{-i+1}p_xu^{-k})u=u (u^{-i}p_xu^{-k+1})\]
and Lemma \ref{conjugate} we see that there are words $\sigma$ and $\tau$ such that $\sigma\tau$ is primitive and
\begin{align*}
(u^{-i+1}p_xu^{-k}) \in (\sigma\tau)^m, & & u =(\sigma\tau)^{\ell}\sigma, & & u^{-i+1}p_xu^{-k}\in (\tau\sigma)^m,
\end{align*}
for some $m\geq 1$ and $\e \geq 0$.
Then we have
\begin{align}\label{**}
&u = (\sigma\tau)^\e\sigma, & & p_x = u^{i}(\tau\sigma)^m u^{k-1}=u^{i-1}(\sigma\tau)^m u^k, \tag{$**$}
\end{align}
for some $m\geq 1$ and $\e \geq 0$. 
}

\begin{enumerate}
\item[(B)]{Let $p_y$ and $p_v$ be conjugate. \smallskip \\
Let $\alpha$ and $\beta$ be such that $p_y =\alpha\beta$ and $p_v =\beta\alpha$. Since $x^ip_y$ is a prefix of $u^ip_v^+$, we can see that $u^{-i}x^i \alpha \beta \pref \beta (\alpha \beta)^+$. From Lemma \ref{cor0} we infer that and $u^{-i}x^i \in \beta (\alpha \beta)^*$. Similarly, by the mirror symmetry,  $p_yx^k \suf p_v^+u^k$ yields that  $x^k u^{-k}\in  (\alpha \beta)^* \alpha$. Then 
\[x^{i+k}=u^ip_v^n u^k,\]
for some $n \geq 1$. From $|v|>|y|$, it follows that $|v| \geq |y|+|p_v|$ and, consequently,
\[(i+k)(|x|-|u|)=j(|v|-|y|) \geq 3|p_v|.\]
Then $n \geq 3$ and we are done by Lemma \ref{cor5}.} \medskip
\item[(C)]{Let $p_x$ and $p_v$ be conjugate.\smallskip \\ 
Let $\alpha$ and $\beta$ be such that $p_x =\alpha\beta$ and $p_v =\beta\alpha$. From \eqref{*} and $i\geq 2$, it follows that $u^ip_v$ is a prefix of $p_x^2$. Then $u^i(\beta\alpha) \pref \alpha (\beta \alpha)^+$ and  Lemma \ref{cor0} yields that $u^i \in \alpha (\beta \alpha)^*$. From $i|u|<|p_x|$, it follows $u^i=\alpha$. Since $p_x$ is a suffix of $\alpha \beta\alpha u^k=p_x u^{i+k}$ and $u$ is a prefix of $p_x$, we  deduce from Lemma \ref{maxsuf} that $x$ and $u$ commute, case (A).}

\end{enumerate}

We will now discuss separately cases when $|x| \geq |v|$ and $|x| <|v|$.\\

\rozbor{1. Suppose that $|x|\geq |v|$.}
\begin{figure}
\centering
\begin{tikzpicture}[axis/.style={very thick, |->, >=stealth'},scale=0.6]
\def\wordlist{0/b, 1/b, 2/a, 3/b, 4/b, 5/a, 6/b, 7/b}
\tikzset{
 arrowbottom/.style={-, >=latex, shorten >=2pt, shorten <=2pt, bend right=45, thick, dashed},
arrowtop/.style={-, >=latex, shorten >=2pt, shorten <=2pt, bend left=45, thick, dashed}
 }
\draw (0,1) rectangle  node[anchor=center]{$x$} (3,2) ;
\draw (3,1) rectangle  node[anchor=center]{$x$} (6,2) ;
\draw (7.5,1) rectangle  node[anchor=center]{$y^j$} (10.5,2) ;
\draw (12,1) rectangle  node[anchor=center]{$x$} (15,2) ;
\draw (0,0) rectangle  node[anchor=center]{$u^i$} (1.5,1) ;
\draw (1.5,0) rectangle  node[anchor=center]{$v$} (4,1) ;
\draw (11,0) rectangle  node[anchor=center]{$v$} (14,1) ;
\draw (14,0) rectangle  node[anchor=center]{$u^k$} (15,1) ;
\draw[dashed] (6,2) -- (7.5,2);
\draw[dashed] (6,1) -- (7.5,1);
\draw[dashed] (10.5,2) -- (12,2);
\draw[dashed] (10.5,1) -- (11,1);
\draw[dashed] (4,0) -- (11,0);
\fill [fill=none] (6,1) rectangle  node[anchor=center]{$...$} (7.5,2);
\fill [fill=none] (10.5,1) rectangle  node[anchor=center]{$...$} (12,2);
\fill [fill=none] (4,0) rectangle  node[anchor=center]{$...$} (11,1);
\end{tikzpicture}
  \caption{Case $|x| \geq |v|$.}\label{fig:case0}
\end{figure}

 If $i\geq 3$ or $x \neq p_x$, then $(u^{-i}x)x^{i-1}$ is a prefix of $v^j$ that is longer than $|p_x|+|x|$ by \eqref{*}. By the Periodicity lemma, $p_x$ is a conjugate of $p_v$ and we are in case (C). The remaining cases  deal with  $i = k = 2$ and $i = 2$, $k = 1$. 

\rozbor{1a)} First suppose that $i = k = 2$. Since $(u^{-i}x)x$ is a prefix of $v^j$ and $x(xu^{-k})$ is a suffix of $v^j$, we get, by Lemma \ref{cor1}, that $(u^{-i}x)x(xu^{-k})$ commutes with $v$. Then
$$x^3 = u^ip_v^n u^k,$$
for some $n \geq 0$. From $(i+k-1)|u| < |p_x| \leq |x|$ and $|p_v| \leq|v| \leq|x|$ we infer that $n\geq 2$. Therefore, $p_u = p_x$ holds by Lemma \ref{cor5}, and we have case (A).

\rozbor{1b)} Suppose now that $i = 2$ and $k = 1$. {We will have a look at the words $u$ and $x=p_x$ expressed by \eqref{**}. Let $h=(\sigma\tau)^m$ and $h' =(\tau\sigma)^m$. Then  \eqref{**} yields
\begin{align*}
u=(\sigma \tau)^\e\sigma, & & x=u^2h'=uhu.
\end{align*}
\rozbor{1b.i)}  Suppose now that $|p_v| \leq|uh|$. Since $h'uh$ is a prefix of $v^j$ and $uh$ is a suffix of $v^j$, we obtain by Lemma \ref{cor1} that
$h'uh = p_v^n$ . From $|p_v| \leq |uh|$, we infer $n \geq 2$ and, according to Lemma \ref{cor5}, $\sigma$ and $\tau$ commute. Then also $x$ and $u$ commute and we have case (A).\\
\rozbor{1b.ii)} Suppose that $|p_v| > |uh|$. From $|x| \geq |v| \geq |p_v|$, it follows that $p_v = h'uu_1$ for
some prefix $u_1$ of $u$. We can suppose that $u_1$ is a proper prefix of $u$, otherwise $x$ and $v$ are conjugate and we have case (C). Then $u_1h' \pref uh' \pref (\sigma \tau)^+$ and, by Lemma \ref{cor2}, we obtain $u u_1^{-1}\in (\sigma\tau)^+$. Therefore, $u_1 \in (\sigma\tau)^*\sigma$. Since $h \suf p_v$, we can see that  $\sigma\tau \suf \tau\sigma^+ $. Lemma \ref{maxsuf} then implies commutativity of $\sigma$ and $ \tau$. Therefore, the words $x$ and $u$ also commute and we are in case (A).}\\

\rozbor{2. Suppose that $|x| < |v|$ and $i|x| = i|u| + |v|$.}\\
{ From $x \suf v$, we have $x \suf xu^k$. Since $u \pref x$ we deduce from Lemma \ref{maxsuf} that $x$ and $u$ commute, thus we have case (A).}\\

\rozbor{3. Suppose that $|x| < |v|$ and $i|x| > i|u| + |v|$.}

\begin{figure}
\centering
\begin{tikzpicture}[axis/.style={very thick, |->, >=stealth'},scale=0.6]
\def\wordlist{0/b, 1/b, 2/a, 3/b, 4/b, 5/a, 6/b, 7/b}
\tikzset{
 arrowbottom/.style={-, >=latex, shorten >=1pt, shorten <=1pt, bend right=45, thick, densely dotted},
arrowtop/.style={-, >=latex, shorten >=1pt, shorten <=1pt, bend left=45, thick, densely dotted}
 }
\draw[arrowtop] (6,2) to node  [above] {$r$}(7.5,2) ;
\draw (0,1) rectangle  node[anchor=center]{$x^i$} (7.5,2) ;
\draw (7.5,1) rectangle  node[anchor=center]{$y^j$} (10.5,2) ;
\draw (10.5,1) rectangle  node[anchor=center]{$x^k$} (15,2) ;
\draw (0,0) rectangle  node[anchor=center]{$u^i$} (1.5,1) ;
\draw (1.5,0) rectangle  node[anchor=center]{$v$} (6,1) ;
\draw (6,0) rectangle  node[anchor=center]{$v^{j-1}$} (14,1) ;
\draw (14,0) rectangle  node[anchor=center]{$u^k$} (15,1) ;
\end{tikzpicture}
  \caption{Case $|x|< |v|$ and $i|x| >i|u|+|v|$.}\label{fig:case1}
\end{figure}

Let $r$ be a non-empty word such that $u^ivr = x^i$. Notice that $|r|<|p_x|$ otherwise the words $p_x$ and $p_v$ are conjugate and we have case (C). Considering the words $u$ and $p_x$ expressed by \eqref{**}, we can see that $(\tau\sigma)^mu^{k-1}u^i$ is a prefix of $v$ and $u^{i-1}(\sigma\tau)^m$ is a suffix of $v$. Notice also that we have case (A) if $\sigma$ and $\tau$ commute.

\rozbor{3a)} Consider first the special case when $r=u^k$.

\rozbor{3a.i)} If $i=k$, then $v^{j-2}=u^iy^ju^i$. If $j \geq 4$, we have case (B) by Lemma \ref{cor5}. If $j=3$, then the equality $u^ivr=x^i$ implies $x^i=u^{2i}y^ju^{2i}$ and we get case (A) again by Lemma \ref{cor5}. 

\rozbor{3a.ii)} Suppose therefore that $k<i$. Notice that $u=\sigma$, otherwise, from  $\tau\sigma\pref v$ and $u^k=r \pref v$, we get commutativity of $\sigma$ and $\tau$. Therefore, 
\[v \in (\tau \sigma)^{m}\sigma^{k-1} p_x^* \sigma^{i-1}(\sigma\tau)^m.\]
We have 
\[vu^kx^{-k} =vrx^{-k}=u^{-i}x^{i-k}.\]
From $i>k$ and \eqref{*} we get $|u^{-i}x^{i-k}|>0$ and, consequently, $|vu^k|>|x^k|$. Let $v'$ dente the word $vu^kx^{-k}$. Then $v^{j-2}v' =ry^j$, and  $j \geq 3$ together with $|v| >|x|>|u^k|=|r|$ yields that $v'$ is a suffix of $y^j$. According to \eqref{**}, $v'=u^{-i}x^{i-k} \in (\tau \sigma)^{m}\sigma^{k-1} p_x^*$. Then, $\sigma^{k}$ is a suffix of $y^j$ and we have
$$(\sigma^ky\sigma^{-k})^j =\sigma^ky^j\sigma^{-k}=v^{j-2} v'\sigma^{-k}.$$
This is a point where Lemma \ref{cor5} turns out to be extremely useful. Direct inspection yields that $v^{j-2}v' \sigma^{-k}$ is not a $j$th power of a word from $\{\sigma,\tau\}^*$. One can verify, for example, that the expression of $v^{j-2}v'\sigma^{-k}$ in terms of $\sigma$ and $\tau$ contains exactly $j-2$ occurrences of $\tau^2$. Therefore, Lemma \ref{cor5} yields that $\sigma$ and $\tau$ commute, a contradiction.

\rozbor{3b)} We first show that  $r=u^k$ holds if $k\geq 2$. Indeed,  if $k\geq 2$ then $u^{k}p_xu^{-k}$ is a suffix of $v$ and, consequently, $u^{k}p_xu^{-k} r$ is a suffix of $x^i$. Since $u^{k}p_xu^{-k} u^k$ is also a suffix of $x^i$, we can use Lemma \ref{cor3} and get commutativity of $x$ with one of the words $u^{-k}r$ or $r^{-1}u^{k}$. From $|r|<|p_x|$ and $|u^k|<|p_x|$, we get $r = u^k$.

\rozbor{3c)} Suppose  that $k=1$ and $r \neq u$. \\
\rozbor{3c.i)} If $|r|<|u|$, then $r$ is a suffix of $u$ and $|xr^{-1}u|>|x|$. Since $xr^{-1} \suf v$ and $k=1$, the word $x=xr^{-1}r$ is a suffix of $xr^{-1}u$. Therefore, $xr^{-1}$ is  a suffix of $(ur^{-1})^+$. Since $u^2 \pref x$ and $|xr^{-1}| \geq |u|+(|u|-|r|)$, the Periodicity lemma implies that the primitive root of $ur^{-1}$ is a conjugate of $p_u$. But since $p_u$ is prefix comparable with $ur^{-1}$, we obtain that $ur^{-1} \in p_u^+$. Then also $r \in p_u^+$ and $xr^{-1} \in p_u^+$.  Consequently, $x$ and $u$ commute, and we have case (A). \\
\rozbor{3c.ii)} Suppose therefore that $|r| > |u|$.  Then $u$ is a suffix of $r$. Since $r$ is a suffix of $p_x$ and $p_x=u^{i}(\tau \sigma)^m$, the word $r$ is a suffix of $u^{i}(\tau \sigma)^m$. From $|v|>|x|$ we obtain  $u^{-i}xu^i \pref v$. Consequently, from $p_x=u^{i}(\tau \sigma)^m$ and $r \pref v$, it follows that $r$ is a prefix of  $(\tau \sigma)^mu^i$.

Consider first the special case when $r \in (\tau \sigma)^m p_u^*$. If $r \in (\tau \sigma)^m p_u^+$, then $r \suf u^{i}(\tau \sigma)^m$ yields that $(\tau\sigma)^m$ and $u$ commute by Lemma \ref{maxsuf}.  Consequently, $\sigma$ and $\tau$ commute, and we have case (A). Therefore, $r=(\tau \sigma)^m$, $p_x=u^ir$ and $v=u^{-i}x^ir^{-1} \in (ru^{i})^+$. We have proved that $x$ and $v$ have conjugate primitive roots, which yields case (C). Consider now the general case.

If $m \leq \e$, then $(\tau \sigma)^m$ is a suffix of $u$. Since $r$ is a prefix of  $(\tau \sigma)^mu^i$, and $u \suf r$, we get from Lemma \ref{cor0} the case $r \in (\tau \sigma)^m p_u^*$. 

Suppose that $m > \e$. Then $u$ is a suffix of $(\tau\sigma)^m$. Let $s$ denote the word $(\tau\sigma)^mu^{-1}=(\tau\sigma)^{m-\e-1}\tau$.

If $|r|\geq|(\tau\sigma)^m|$, then $r=s'su$ for some $s'$. From $r \pref (\tau \sigma)^mu^i$, it follows that $s'su$ is a prefix of $su^{i+1}$. Lemma \ref{cor0cor} then yields $s's\in sp_u^*$. Therefore $r\in sup_u^*$ and from $su=(\tau\sigma)^m$, we have the case $r \in (\tau \sigma)^m p_u^*$.

Let $|r|<|(\tau\sigma)^m|$. From $|r|>|u|$ and $(\tau\sigma)^m=su$, we obtain that there are words $s_1,s_2$ such that $s=s_1s_2$, $r=s_2u\pref v$ and $s_1\suf v$. Since $s$ is both a prefix and a suffix of $v$, Lemma \ref{mother} implies that $s_1$ and $s_2$ have the same primitive root, namely $p_s$.

Note that $p_x=u^is u$. We now have
\[u^is_2s_1=u^{i}s \suf v\suf x^ir^{-1} \suf (u^isu)^{i-1}u^is_1.\]
From $i \geq 2$, it follows that $u^{i} s_2$ is a suffix of $(u^isu)^{i-1}u^i$ for some $n \geq 1$. Lemma \ref{maxsuf} then yields commutativity of $s$ and $u$. Hence, words $x$ and $u$ also commute and we are in case (A).\\

\rozbor{4. Suppose now that $|x| < |v|$ and $i|x| < i|u| + |v|$.}

\begin{figure}
\centering
\begin{tikzpicture}[axis/.style={very thick, |->, >=stealth'},scale=0.6]
\def\wordlist{0/b, 1/b, 2/a, 3/b, 4/b, 5/a, 6/b, 7/b}
\tikzset{
 arrowbottom/.style={-, >=latex, shorten >=1pt, shorten <=1pt, bend right=45, thick, densely dotted},
arrowtop/.style={-, >=latex, shorten >=1pt, shorten <=1pt, bend left=45, thick, densely dotted}
 }
\draw[arrowtop] (6.8,2) to node  [above] {$\alpha$}(7.3,2) ;
\draw[arrowtop] (7.3,2) to node  [above] {$\beta$}(8.5,2) ;
\draw[arrowtop] (8.5,2) to node  [above] {$\gamma$}(9.1,2) ;
\draw[arrowtop] (9.1,2) to node  [above] {$\alpha$}(9.6,2) ;
\draw[arrowtop] (9.6,2) to node  [above] {$\beta$}(10.8,2) ;
\draw[arrowtop] (10.8,2) to node  [above] {$\gamma$}(11.4,2) ;
\draw[arrowtop] (11.4,2) to node  [above] {$\alpha$}(11.9,2) ;
\draw[arrowtop] (11.9,2) to node  [above] {$\beta$}(13.1,2) ;
\draw[arrowtop] (13.1,2) to node  [above] {$\gamma$}(13.7,2) ;
\draw (0,1) rectangle  node[anchor=center]{$x^i$} (6.8,2) ;
\draw (6.8,1) rectangle  node[anchor=center]{$y$} (9.1,2) ;
\draw (9.1,1) rectangle  node[anchor=center]{$y$} (11.4,2) ;
\draw (11.4,1) rectangle  node[anchor=center]{$y$} (13.7,2) ;
\draw (13.7,1) rectangle  node[anchor=center]{$x^k$} (20,2) ;
\draw (0,0) rectangle  node[anchor=center]{$u^i$} (1.5,1) ;
\draw (1.5,0) rectangle  node[anchor=center]{$v$} (7.3,1) ;
\draw (7.3,0) rectangle  node[anchor=center]{$v$} (13.1,1) ;
\draw (13.1,0) rectangle  node[anchor=center]{$v$} (19,1) ;
\draw (19,0) rectangle  node[anchor=center]{$u^k$} (20,1) ;
\end{tikzpicture}
  \caption{Case $|x| < |v|$ and $i|x| <i|u|+|v|$.}\label{fig:case2}
\end{figure}
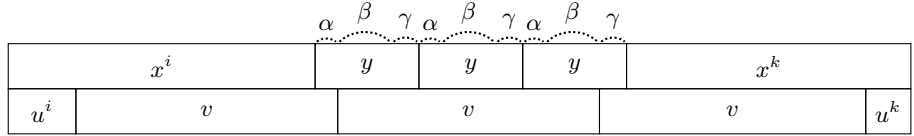

First notice that in this case also $k|x| < k|u|+|v|$. If $j|y| \geq |v|+|p_y|$, then, by the Periodicity lemma, $p_v$ and $p_y$ are conjugate, and theorem holds by (B). Assume that $j|y| < |v| + |p_y|$. Then, since $i|x| < i|u| + |v|$ and $k|x| < k|u| + |v| $, we can see that $j = 3$ and there are non-empty words $\alpha$, $\beta$ and $\gamma$ for which $y = \alpha\beta\gamma$ and $v = (\beta\gamma)(\alpha\beta\gamma)(\alpha\beta)$, with $|\alpha\gamma| < |p_y|$.

\rozbor{4a)} Suppose first that $|u^i\gamma| \leq|x|$. Notice that also $|\alpha u^k| \leq|x|$ since $k \leq i$ and $|\gamma| = (i-k)(|x|-|u|)+|\alpha|$. Then $|\gamma x| \leq |v|$ and  $u^i\gamma x$ is a prefix of $x^2$. Therefore, by Lemma \ref{cor0}, $u^i\gamma$ commutes with $x$. We obtain the following equalities:
\begin{align*}
&v = \gamma p_x^n\alpha, & & y^j = \alpha v\gamma = (\alpha\gamma)p_x^{n}(\alpha\gamma),
\end{align*}
where $n \geq 1$. If $n \geq 2$, then $x$ and $y$ commute by Lemma \ref{cor5}.  If $n=1$, then $p_x =x$ and $i=2$. Since $\gamma x^k =vu^k =\gamma x\alpha u^k$ and $|\alpha u^k| \leq|x|$, also $k=2$ and $\alpha u^k =x$. Then $|\alpha|=|\gamma|$ and $u^2\gamma=x=\alpha u^2$. If $|u|\geq|\gamma|$, then $u$ and $\gamma$ commute, a contradiction with $p_x = x$. Therefore, $ |x| < 3|\gamma|$ and $ |v| = |\gamma x\alpha| < 5|\gamma|$. Since $ \gamma$ is a suffix of $x$ and $\alpha$ is a prefix of $x$, $(\gamma\alpha\beta)^3\gamma\alpha$ is a factor of $v^3$ longer than $|y| + |v|$. Therefore, by the Periodicity lemma, words $y$ and $v$ are conjugate, and we have case (B).\\

\rozbor{4b)} Suppose that $|u^i\gamma| > |x|$, denote $z = x^{-1}u^i\gamma$ and $z' = \gamma^{-1}v\alpha^{-1} = x^ku^{-k}\alpha^{-1}$. From
$$|y| + |\gamma| + |\alpha| < |v| = |\gamma z'\alpha|,$$
we deduce $|y| <|z'|$ . Since $x^{i-1} = zz' $ and $z'$ is  a prefix of $x^k$, the word $zz' $ has a period $|z| < |\gamma|$. Since $zz'$ is a factor of $v$ greater than $|z| + |y|$ and $v$ has a period $|p_y|$, the Periodicity lemma implies $|p_y| \leq|z| < |\gamma|$, a contradiction with $|\gamma| < |p_y|$.
\qed

\end{proof}

\section{Conclusion}

The minimal bounds for $i$, $j$, $k$ in the previous theorem are optimal. 
{
{ In case that $i=k$ and $j$ is even, Eq. \eqref{eqn:rovnice1} splits into two separate equations, which have a solution if and only if either $i=k$ and $j=2$, or $i=k=1$, see \cite{karhumaki}.} 

Apart from these solutions, we can  find non-periodic solutions also in case that $i \neq k$. Namely, for $j=2$ and $i=k+1$, we have
\begin{align*}
&x = \alpha^{2k+1}(\beta\alpha^k)^2, & &u = \alpha, \\
&y = \beta\alpha^k, & &v = (\alpha ^k\beta)^2(\alpha^{3k+1}\beta\alpha^k\beta )^k.
\end{align*}
So far this seems to be the only situation when the equation
\begin{align}\label{dvojka}
x^iy^2x^k = u^iv^2u^k
\end{align}
with  $i > k$ has a non-periodic solution. We conjecture that if $|i-k| \geq 2$, then Eq. \eqref{dvojka} has only periodic solutions. 

If $i=k=1$ and $j$ is odd, then Eq. \eqref{eqn:rovnice1} has several non-periodic solutions, for example:
\begin{align*}
& x = \alpha\beta\alpha, & & u = \alpha, \\
& y = \gamma, & &v = \alpha\gamma^j\alpha,
\end{align*}
where $\beta^2=v^{j-1}$.}


\bibliographystyle{splncs03}
\bibliography{bibliography1}

\end{document}